%% file: partialPreferences.tex
\documentclass[submission,copyright,creativecommons]{eptcs}

\usepackage{xspace} 
\usepackage{mathrsfs} 
\usepackage{makeidx}
\usepackage{url}
\usepackage{latexsym} 
\usepackage{amsmath}
\usepackage{amssymb}
\usepackage{amsfonts}
\usepackage{enumerate}

\providecommand{\event}{SR 2014} 
\usepackage{breakurl}             

\title{Partial Preferences for Mediated Bargaining}
\author{
Piero A. Bonatti
\and
Marco Faella
\institute{Dept. of Electrical Engineering and Information Technologies\\
Universit\`a di Napoli ``Federico II''\\
Italy}
\and
Luigi Sauro
}
\def\titlerunning{Partial Preferences}
\def\authorrunning{P.A. Bonatti, M. Faella \& L. Sauro}

\input{macro}

\begin{document} 

\maketitle
\input{introduction}
\input{section1}
\input{section2}

\input{section3}
\input{conclusions}
\bibliography{partialAuctions}
\bibliographystyle{eptcs}
\end{document}

%% file: macro.tex
\newcommand{\hide}[1]{}
\newcommand{\remove}[1]{}

\newcommand{\conv}[3]{\ensuremath{\langle #1, #2, #3 \rangle}\xspace}

\newcommand{\A}{\ensuremath{\mathcal{A}}\xspace}

\newcommand{\supp}{\ensuremath{\mathit{supp}}\xspace}
\newcommand{\incomp}{\not\lessgtr}

\newcommand{\prob}{\ensuremath{\Delta}\xspace}

\newtheorem{theorem}{Theorem}
\newtheorem{lemma}{Lemma}

\newenvironment{proof}
{\begin{trivlist} \item[] {\bf Proof.}}%
{\qed \end{trivlist}}

\newcommand{\qed}{\hspace*{\fill} \rule{1.1mm}{2.2mm}}

%% file: introduction.tex
\begin{abstract}
In this work we generalize standard Decision Theory by assuming that two outcomes can also be incomparable.
Two motivating scenarios show how incomparability may be helpful to represent 
those situations where, due to lack of information, the decision maker would 
like to maintain different options \emph{alive} and defer the final decision.
In particular, a new axiomatization is given which turns out to be a weakening of the classical
set of axioms used in Decision Theory.
Preliminary results show how preferences involving complex distributions
are related to judgments on single alternatives.  
\end{abstract}

\section{Introduction}

In his pioneering work on Decision Theory \cite{Savage},
when delineating the fundamental properties of a preference relation $\prec$,
Savage makes the following point:
given two potential outcomes $f$ and $g$, 
it cannot be the case that $f\prec g$ and $g\prec f$ at the same time.
Clearly, this is logically equivalent to saying that 
either $f\not\prec g$ or $g\not\prec f$, which leads to three possible cases:
\emph{(i)} $f\not\prec g$ and $g\prec f$, 
\emph{(ii)} $f\prec g$ and $g\not\prec f$, or
\emph{(iii)} $f\not\prec g$ and $g\not\prec f$.
Then, he postulates that these three cases are the only possible judgments concerning $f$ and $g$. 
In particular, the last case ($f\not\prec g$ and $g\not\prec f$) allegedly 
implies that $f$ and $g$ are equivalent in the sense that in any situation wherein these are the only two possible options, the decision maker does not mind delegating to coin flipping. 
Consequently, in classical Decision Theory (\emph{CDT}) a very fundamental property of a  
preference relation is its totality. 

From the theory's very start, the hidden assumptions underlying this model of an \emph{economic man} raised some criticisms, one of the most influential of which being due to Simon: 
\begin{quote}
``This man is assumed to have knowledge of the relevant aspects of his environment which, if not absolutely complete, is at least impressively clear and voluminous. He is assumed also to have a 
well-organized and stable system of preferences, and a skill in computation that enables him to calculate, for the alternative courses of action that are available to him, which of these will permit him to reach the highest attainable point on his preference scale''  \cite{Simon55}.
\end{quote}
 
In recent years,  the massive development of e-commerce services makes Simon's criticisms even more cogent and the classical viewpoint on the economic man more and more idealistic.
Often preferences result from complex trade-offs between different attributes (functionalities, cost, Quality of Service (QoS), information disclosure risks, etc.) sometimes the user has only a vague idea of. Moreover, in some cases the user actually consists of a group of persons where internal debate 
does not easily end up with a total preference.
Finally, from a computer-science perspective, 
our aim could be to develop a software agent that acts in an electronic market on behalf of 
a real user. As we discuss in Section~\ref{sec:scenarios}, 
even if the user conforms with the classical economic man, her preference relation 
could be so complex that it could not be entirely and efficiently injected into the software agent.

Differently from Simon, 
who moved towards a problem-solving perspective,
in this work we challenge CDT on its playground. In particular, we provide an alternative axiomatization 
where two outcomes, due to the lack of information or an irreducible heterogeneity of the attributes involved, can also be incomparable. 
In Section 2,  we show two motivating scenarios where considering  preferences as incomplete seems to be appropriate.
In Section 3,  we introduce the new axiomatization by emphasizing which axioms of CDT remain 
the same and which axioms should be replaced. The resulting theory, that we call Partial Decision Theory, consists in  a weakening of CDT, that is, 
all the properties it satisfies are satisfied by CDT as well (but not vice versa).  
In Section 3, we show some general features of the new axiomatization; 
in particular, we argue that the proposed framework is not too weak, as it retains several 
desirable properties of CDT. Conclusions and future works end the paper.

%% file: section1.tex
\section{Motivating Scenarios}\label{sec:scenarios}


\vspace*{-.8em}
In this section we introduce two scenarios where partial preferences seem to provide a more natural way to describe a decision maker, or a software 
agent behaving on its behalf, than total preferences.  

In the first scenario, Bob wants to learn to play the piano and posts a request on a consumer-to-consumer social network. Soon after, he receives offers from two musicians, Carl and Mary. Carl provides two options: a 5 people class for 15 dollars per person or a one-to-one class for 35 dollars. Mary offers two similar options: a 3 people class for 20 dollars per person or a one-to-one class for 40 dollars.
Furthermore, they both offer a trial lesson.
Regarding Carl's options, Bob thinks that 5 people are too many for a class,
thus he prefers the one-to-one option. 
On the contrary, he judges the price difference between Mary's options somewhat excessive,
hence he prefers the 3 people class.  
If someone asks Bob \emph{``Do you prefer Mary's 3 people class or Carl's one-to-one class?''},
Bob will probably answer \emph{``I do not know, 
I first have to attend the trial lessons''}.
Notice that this is different from saying that the two options are equivalent, because in that case  Bob would simply flip a coin and choose one of them. On the contrary, it is more natural to think that these options are initially \emph{incomparable} 
and Bob will use the trial lessons to disambiguate them.   

More generally,  in absence of complete information it might be difficult for an individual to figure out a coherent total order over the bids and choose in a single step one of them.   
On the contrary, making a decision can be viewed as a multiple step process where offers are initially filtered according
to a partial preference relation. Then,  depending on the resulting offers, an individual can acquire further information and possibly rank them.   

In the second scenario, Alice's father has finally agreed to buy her a smartphone, 
and now she is browsing Ebay for possible offers. Unfortunately, the list is huge and patience is not Alice's forte. So, she would like to be assisted by a software agent to filter out undesired options. 
The software agent accepts constraints such as maximum cost and size, color restrictions, etc., and also a preference relation as a total order over bids. Then, according to the specified preference relation, the agent returns  the best offer. 
Clearly, Alice's desires are influenced by several attributes such as operating system, color, weight, brand, and so on. For instance, she has a preference over operating systems in the following decreasing order: OS1, OS2 and OS3; over colors: blue, red, black, white; and over brands: Brand1,
Brand2, Brand3. Furthermore, out of benevolence for her father, given a specific model the cheaper the better. However, such preferences over single attributes do not constitute a total order. Moreover, Alice cannot establish a priority over attributes, for instance she prefers a Brand3 phone 
with operating system OS1 to a Brand2 one with operating system OS2, but she prefers a Brand1 OS2 phone to a Brand2 OS1 one. 
Alice soon realizes that providing a total order to the software agent is frustrating and requires about the same effort as comparing all the offers by herself.
This scenario reveals the following issue: 
in designing a software agent that behaves on behalf of real users,  we have to take into account how users can \emph{instruct} the agent about their own preferences. In electronic markets where the number of offers can be huge, it could be unfeasible to transfer an 
exact representation of users' desires into a software agent.
In this case, the agent should make do with an approximate representation of users' desires as a partial order and return a restricted list of choices from which the user can select the preferred one. As a further advantage, the user retains the ability of applying unforeseen, situation-specific knowledge and preferences that had not been formalized in advance.

%% file: section2.tex
\section{Partial Preferences}

Let $\prob(\A)$ be the class of all probability distributions over a countable set of alternatives $\A$. Given two probability distributions $f,g\in\prob(\A)$ and $\alpha\in[0,1]$, 
we denote by $\conv{\alpha}{f}{g}$ the convex combination of $f$ and $g$
such that $\conv{\alpha}{f}{g}(a) = \alpha f(a) + (1-\alpha)g(a)$, for all $a\in\A$.
Moreover, for an alternative $a\in\A$, $[a]$ denotes the degenerate distribution that assigns probability
1 to $a$.  

A preference relation $\preceq$ is 
a binary relation on $\prob(\A)$, subject to the following classical Decision Theory axioms: \footnote{Here, we borrow the formulation presented in~\cite{Myerson}.}
\begin{enumerate} 
\item\label{ca:totality} $f\preceq g$ or $g\preceq f$ \emph{(totality)};
\item\label{ca:transitivity} if $f\preceq g$ and $g\preceq h$, then $f\preceq h$ \emph{(transitivity)};
\item\label{ca:prec1} if $f\prec g$ and $0\le \alpha <\beta\le 1$, then
                      $\conv{\beta}{f}{g} \prec \conv{\alpha}{f}{g}$; 
\item\label{ca:preceq1} if $f_1 \preceq g_1$, $f_2\preceq g_2$, and $0\le \alpha\le 1$, 
then $\conv{\alpha}{f_1}{f_2} \preceq \conv{\alpha}{g_1}{g_2}$; 
\item\label{ca:pref_dom} if $f_1\prec g_1$, $f_2 \preceq g_2$, and $0\le \alpha\le 1$, 
then $\conv{\alpha}{f_1}{f_2} \prec \conv{\alpha}{g_1}{g_2}$;
\end{enumerate} 
where, as usual, 
$f\prec g$ means that $f\preceq g$ and $g\not\preceq f$.
 
Notice that the first two axioms force $\preceq$ to be a total (hence reflexive) transitive relation, 
i.e., a total preorder (also called non-strict weak order). 
As shown by the previous scenarios, we advocate that in several contexts some outcomes may be incomparable,
meaning that $\preceq$ should be modeled as a (possibly partial) preorder. 
For this reason, we weaken the totality axiom in favor of  
one which requires reflexivity only:  
\begin{enumerate}[ 1'.]
\item\label{ax:reflexivity} $f\preceq f$;
\end{enumerate}  
Having allowed for incomparable distributions, at first look it may seem that the deal is done.
However, the obtained theory is so weak that it contemplates unrealistic preferences. 
The problem is that the previous axioms do not say anything about incomparable distributions which, once combined,
can be then freely judged. On the contrary, it is natural to think that, to some extent, the incomparability between distributions persists
also when they are combined.

Assume for example that $f$ and $g$ are incomparable and let $0\le \alpha<\beta\le 1$.
According to the axioms above, it is possible that $\conv{\alpha}{f}{g} \prec \conv{\beta}{f}{g}$.
This looks inappropriate: given that I cannot compare $f$ and $g$,
why should I strictly prefer one combination of $f$ and $g$ over another?
This leads to a further axiom:
\begin{enumerate}
\setcounter{enumi}{5}
\item\label{ax:upper1} if $0\le \alpha\le 1$,
      and $\conv{\alpha}{f_1}{f_2} \prec \conv{\alpha}{g_1}{g_2}$,
      then there exist $j,k \in\{1,2\}$ such that $f_j \prec g_k$.
\end{enumerate} 
Intuitively, a distribution $f$ of the type $\conv{\alpha}{f_1}{f_2}$ can be seen as a 
random choice (a.k.a. a \emph{lottery})
which picks $f_1$ with probability $\alpha$ and $f_2$ with probability $1-\alpha$.
Comparing $f$ with another distribution $g$ of the type $\conv{\alpha}{g_1}{g_2}$
encompasses comparing four possible draws: $(f_1,g_1)$, $(f_1,g_2)$, $(f_2,g_1)$, and $(f_2,g_2)$.\footnote{With probabilities $\alpha^2$, 
$\alpha(1-\alpha)$, $\alpha(1-\alpha)$ and $(1-\alpha)^2$, respectively.}
If there is no draw in which the second component is strictly better than the second, 
Axiom~\ref{ax:upper1} requires that $g$ is not strictly preferred to $f$.
Notice that one could easily come up with more stringent conditions on the persistence of
incomparability, for instance by requiring that a majority of draws favors the second component
over the first one. We instead propose a rather weak requirement, in the form of Axiom~\ref{ax:upper1},
which supports a wide range of preference relations, while still ensuring a number of interesting
properties, which are the subject of Section~\ref{sec:props}. 
 
We call \emph{Partial Decision Theory} (PDT) the new set of axioms 1'-\ref{ax:upper1}. 
Notice that Axiom~\ref{ax:upper1} can be easily derived in classical Decision Theory, 
consequently PDT is a weakening of classical Decision Theory,
in the sense that all preference relations satisfying classical Decision Theory also satisfy PDT.
The converse does not hold, 
as witnessed by the ``empty'' preference relation, i.e., the relation
that considers incomparable all distinct distributions.

As seen above, Axiom~\ref{ax:upper1} has been motivated by analyzing which preferences between 
$f=\conv{\alpha}{f_1}{f_2}$ and $g=\conv{\alpha}{g_1}{g_2}$ are admissible on the basis of the preferences on the possible draws $(f_1,g_1)$, $(f_1,g_2)$, $(f_2,g_1)$, and $(f_2,g_2)$.
In Table \ref{tab:table} we perform such an analysis extensively.  
In particular, for each entry, the left-hand side  $\bowtie_1\:\bowtie_2\:\bowtie_3\:\bowtie_4$, with $\bowtie\in\{\sim,\prec,\succ,\incomp\}$, 
is a consistent combination  $f_1\bowtie_1 g_1$, $f_1\bowtie_2 g_2$, $f_2 \bowtie_3 g_1$, and $f_2 \bowtie_4 g_2$,
whereas the right-hand side shows which preference relations between $f$ and $g$ are consistent with PDT. For example, the first entry is the case 
$f_1\sim g_1$, $f_1 \sim g_2$, $f_2\sim g_1$, and $f_2 \sim g_2$, then according with Axiom~\ref{ca:preceq1}, $f\sim g$ is the only possibility. 
Conversely, in some other cases 
(e.g. $\prec\incomp\succ\incomp$) no axiom can be applied, consequently $f$ can be in any relationship with $g$.
 
Table \ref{tab:table} provides a close look on how PDT behaves and hence it can be a good starting point to debate whether and how it can be extended or modified. 
For example, notice that in some cases $f$ and $g$ are comparable even if some of the underlying draws are not 
(e.g. due to  Axiom~\ref{ca:pref_dom}, $\prec\incomp\incomp\sim$ results in $f\prec g$). 
Somewhat conversely, $f$ and $g$ can be incomparable even if all the underlying draws are comparable (e.g. $\prec\succ\prec\succ$ admits $f\incomp g$). Finally, 
$f$ and $g$ are never forced to be incomparable, even if $f_1\incomp g_1$, $f_1\incomp g_2$, $f_2\incomp g_2$, and $f_2\incomp g_2$.       
 
\begin{table}
\begin{center}
{\small
\begin{tabular}{|p{2cm}r|p{2cm}r|p{2cm}r|p{2cm}r|}
$\sim$$\sim$$\sim$$\sim$ &   $\sim$ & $\sim$$\sim$$\prec$$\prec$ &   $\prec$ & $\sim$$\sim$$\succ$$\succ$ &   $\succ$ & $\sim$$\sim$$\incomp$$\incomp$ &   $\sim$  $\incomp$ \\ 
$\sim$$\prec$$\sim$$\prec$ &   $\prec$ & $\sim$$\prec$$\prec$$\prec$ &   $\prec$ & $\sim$$\prec$$\succ$$\sim$ &   $\sim$ & $\sim$$\prec$$\succ$$\prec$ &   $\prec$ \\ 
$\sim$$\prec$$\succ$$\succ$ &   $\succ$ & $\sim$$\prec$$\succ$$\incomp$ &   $\sim$  $\prec$  $\succ$  $\incomp$ & $\sim$$\prec$$\incomp$$\prec$ &   $\prec$ & $\sim$$\prec$$\incomp$$\incomp$ &   $\sim$  $\prec$  $\incomp$ \\ 
$\sim$$\succ$$\sim$$\succ$ &   $\succ$ & $\sim$$\succ$$\prec$$\sim$ &   $\sim$ & $\sim$$\succ$$\prec$$\prec$ &   $\prec$ & $\sim$$\succ$$\prec$$\succ$ &   $\succ$ \\ 
$\sim$$\succ$$\prec$$\incomp$ &   $\sim$  $\prec$  $\succ$  $\incomp$ & $\sim$$\succ$$\succ$$\succ$ &   $\succ$ & $\sim$$\succ$$\incomp$$\succ$ &   $\succ$ & $\sim$$\succ$$\incomp$$\incomp$ &   $\sim$  $\succ$  $\incomp$ \\ 
$\sim$$\incomp$$\sim$$\incomp$ &   $\sim$  $\incomp$ & $\sim$$\incomp$$\prec$$\prec$ &   $\prec$ & $\sim$$\incomp$$\prec$$\incomp$ &   $\sim$  $\prec$  $\incomp$ & $\sim$$\incomp$$\succ$$\succ$ &   $\succ$ \\ 
$\sim$$\incomp$$\succ$$\incomp$ &   $\sim$  $\succ$  $\incomp$ & $\sim$$\incomp$$\incomp$$\sim$ &   $\sim$ & $\sim$$\incomp$$\incomp$$\prec$ &   $\prec$ & $\sim$$\incomp$$\incomp$$\succ$ &   $\succ$ \\ 
$\sim$$\incomp$$\incomp$$\incomp$ &   $\sim$  $\incomp$ & $\prec$$\sim$$\sim$$\succ$ &   $\sim$  $\prec$  $\succ$  $\incomp$ & $\prec$$\sim$$\prec$$\sim$ &   $\prec$ & $\prec$$\sim$$\prec$$\prec$ &   $\prec$ \\ 
$\prec$$\sim$$\prec$$\succ$ &   $\sim$  $\prec$  $\succ$  $\incomp$ & $\prec$$\sim$$\prec$$\incomp$ &   $\sim$  $\prec$  $\incomp$ & $\prec$$\sim$$\succ$$\succ$ &   $\sim$  $\prec$  $\succ$  $\incomp$ & $\prec$$\sim$$\incomp$$\succ$ &   $\sim$  $\prec$  $\succ$  $\incomp$ \\ 
$\prec$$\sim$$\incomp$$\incomp$ &   $\sim$  $\prec$  $\incomp$ & $\prec$$\prec$$\sim$$\sim$ &   $\prec$ & $\prec$$\prec$$\sim$$\prec$ &   $\prec$ & $\prec$$\prec$$\sim$$\succ$ &   $\sim$  $\prec$  $\succ$  $\incomp$ \\ 
$\prec$$\prec$$\sim$$\incomp$ &   $\sim$  $\prec$  $\incomp$ & $\prec$$\prec$$\prec$$\sim$ &   $\prec$ & $\prec$$\prec$$\prec$$\prec$ &   $\prec$ & $\prec$$\prec$$\prec$$\succ$ &   $\sim$  $\prec$  $\succ$  $\incomp$ \\ 
$\prec$$\prec$$\prec$$\incomp$ &   $\sim$  $\prec$  $\incomp$ & $\prec$$\prec$$\succ$$\sim$ &   $\prec$ & $\prec$$\prec$$\succ$$\prec$ &   $\prec$ & $\prec$$\prec$$\succ$$\succ$ &   $\sim$  $\prec$  $\succ$  $\incomp$ \\ 
$\prec$$\prec$$\succ$$\incomp$ &   $\sim$  $\prec$  $\succ$  $\incomp$ & $\prec$$\prec$$\incomp$$\sim$ &   $\prec$ & $\prec$$\prec$$\incomp$$\prec$ &   $\prec$ & $\prec$$\prec$$\incomp$$\succ$ &   $\sim$  $\prec$  $\succ$  $\incomp$ \\ 
$\prec$$\prec$$\incomp$$\incomp$ &   $\sim$  $\prec$  $\incomp$ & $\prec$$\succ$$\sim$$\succ$ &   $\sim$  $\prec$  $\succ$  $\incomp$ & $\prec$$\succ$$\prec$$\sim$ &   $\prec$ & $\prec$$\succ$$\prec$$\prec$ &   $\prec$ \\ 
$\prec$$\succ$$\prec$$\succ$ &   $\sim$  $\prec$  $\succ$  $\incomp$ & $\prec$$\succ$$\prec$$\incomp$ &   $\sim$  $\prec$  $\succ$  $\incomp$ & $\prec$$\succ$$\succ$$\succ$ &   $\sim$  $\prec$  $\succ$  $\incomp$ & $\prec$$\succ$$\incomp$$\succ$ &   $\sim$  $\prec$  $\succ$  $\incomp$ \\ 
$\prec$$\succ$$\incomp$$\incomp$ &   $\sim$  $\prec$  $\succ$  $\incomp$ & $\prec$$\incomp$$\sim$$\succ$ &   $\sim$  $\prec$  $\succ$  $\incomp$ & $\prec$$\incomp$$\sim$$\incomp$ &   $\sim$  $\prec$  $\incomp$ & $\prec$$\incomp$$\prec$$\sim$ &   $\prec$ \\ 
$\prec$$\incomp$$\prec$$\prec$ &   $\prec$ & $\prec$$\incomp$$\prec$$\succ$ &   $\sim$  $\prec$  $\succ$  $\incomp$ & $\prec$$\incomp$$\prec$$\incomp$ &   $\sim$  $\prec$  $\incomp$ & $\prec$$\incomp$$\succ$$\succ$ &   $\sim$  $\prec$  $\succ$  $\incomp$ \\ 
$\prec$$\incomp$$\succ$$\incomp$ &   $\sim$  $\prec$  $\succ$  $\incomp$ & $\prec$$\incomp$$\incomp$$\sim$ &   $\prec$ & $\prec$$\incomp$$\incomp$$\prec$ &   $\prec$ & $\prec$$\incomp$$\incomp$$\succ$ &   $\sim$  $\prec$  $\succ$  $\incomp$ \\ 
$\prec$$\incomp$$\incomp$$\incomp$ &   $\sim$  $\prec$  $\incomp$ & $\succ$$\sim$$\sim$$\prec$ &   $\sim$  $\prec$  $\succ$  $\incomp$ & $\succ$$\sim$$\prec$$\prec$ &   $\sim$  $\prec$  $\succ$  $\incomp$ & $\succ$$\sim$$\succ$$\sim$ &   $\sim$  $\succ$  $\incomp$ \\ 
$\succ$$\sim$$\succ$$\prec$ &   $\sim$  $\prec$  $\succ$  $\incomp$ & $\succ$$\sim$$\succ$$\succ$ &   $\succ$ & $\succ$$\sim$$\succ$$\incomp$ &   $\sim$  $\succ$  $\incomp$ & $\succ$$\sim$$\incomp$$\prec$ &   $\sim$  $\prec$  $\succ$  $\incomp$ \\ 
$\succ$$\sim$$\incomp$$\incomp$ &   $\sim$  $\succ$  $\incomp$ & $\succ$$\prec$$\sim$$\prec$ &   $\sim$  $\prec$  $\succ$  $\incomp$ & $\succ$$\prec$$\prec$$\prec$ &   $\sim$  $\prec$  $\succ$  $\incomp$ & $\succ$$\prec$$\succ$$\sim$ &   $\sim$  $\prec$  $\succ$  $\incomp$ \\ 
$\succ$$\prec$$\succ$$\prec$ &   $\sim$  $\prec$  $\succ$  $\incomp$ & $\succ$$\prec$$\succ$$\succ$ &   $\succ$ & $\succ$$\prec$$\succ$$\incomp$ &   $\sim$  $\prec$  $\succ$  $\incomp$ & $\succ$$\prec$$\incomp$$\prec$ &   $\sim$  $\prec$  $\succ$  $\incomp$ \\ 
$\succ$$\prec$$\incomp$$\incomp$ &   $\sim$  $\prec$  $\succ$  $\incomp$ & $\succ$$\succ$$\sim$$\sim$ &   $\sim$  $\succ$  $\incomp$ & $\succ$$\succ$$\sim$$\prec$ &   $\sim$  $\prec$  $\succ$  $\incomp$ & $\succ$$\succ$$\sim$$\succ$ &   $\succ$ \\ 
$\succ$$\succ$$\sim$$\incomp$ &   $\sim$  $\succ$  $\incomp$ & $\succ$$\succ$$\prec$$\sim$ &   $\sim$  $\prec$  $\succ$  $\incomp$ & $\succ$$\succ$$\prec$$\prec$ &   $\sim$  $\prec$  $\succ$  $\incomp$ & $\succ$$\succ$$\prec$$\succ$ &   $\succ$ \\ 
$\succ$$\succ$$\prec$$\incomp$ &   $\sim$  $\prec$  $\succ$  $\incomp$ & $\succ$$\succ$$\succ$$\sim$ &   $\sim$  $\succ$  $\incomp$ & $\succ$$\succ$$\succ$$\prec$ &   $\sim$  $\prec$  $\succ$  $\incomp$ & $\succ$$\succ$$\succ$$\succ$ &   $\succ$ \\ 
$\succ$$\succ$$\succ$$\incomp$ &   $\sim$  $\succ$  $\incomp$ & $\succ$$\succ$$\incomp$$\sim$ &   $\sim$  $\succ$  $\incomp$ & $\succ$$\succ$$\incomp$$\prec$ &   $\sim$  $\prec$  $\succ$  $\incomp$ & $\succ$$\succ$$\incomp$$\succ$ &   $\succ$ \\ 
$\succ$$\succ$$\incomp$$\incomp$ &   $\sim$  $\succ$  $\incomp$ & $\succ$$\incomp$$\sim$$\prec$ &   $\sim$  $\prec$  $\succ$  $\incomp$ & $\succ$$\incomp$$\sim$$\incomp$ &   $\sim$  $\succ$  $\incomp$ & $\succ$$\incomp$$\prec$$\prec$ &   $\sim$  $\prec$  $\succ$  $\incomp$ \\ 
$\succ$$\incomp$$\prec$$\incomp$ &   $\sim$  $\prec$  $\succ$  $\incomp$ & $\succ$$\incomp$$\succ$$\sim$ &   $\sim$  $\succ$  $\incomp$ & $\succ$$\incomp$$\succ$$\prec$ &   $\sim$  $\prec$  $\succ$  $\incomp$ & $\succ$$\incomp$$\succ$$\succ$ &   $\succ$ \\ 
$\succ$$\incomp$$\succ$$\incomp$ &   $\sim$  $\succ$  $\incomp$ & $\succ$$\incomp$$\incomp$$\sim$ &   $\sim$  $\succ$  $\incomp$ & $\succ$$\incomp$$\incomp$$\prec$ &   $\sim$  $\prec$  $\succ$  $\incomp$ & $\succ$$\incomp$$\incomp$$\succ$ &   $\succ$ \\ 
$\succ$$\incomp$$\incomp$$\incomp$ &   $\sim$  $\succ$  $\incomp$ & $\incomp$$\sim$$\sim$$\incomp$ &   $\sim$  $\incomp$ & $\incomp$$\sim$$\prec$$\prec$ &   $\sim$  $\prec$  $\incomp$ & $\incomp$$\sim$$\prec$$\incomp$ &   $\sim$  $\prec$  $\incomp$ \\ 
$\incomp$$\sim$$\succ$$\succ$ &   $\sim$  $\succ$  $\incomp$ & $\incomp$$\sim$$\succ$$\incomp$ &   $\sim$  $\succ$  $\incomp$ & $\incomp$$\sim$$\incomp$$\sim$ &   $\sim$  $\incomp$ & $\incomp$$\sim$$\incomp$$\prec$ &   $\sim$  $\prec$  $\incomp$ \\ 
$\incomp$$\sim$$\incomp$$\succ$ &   $\sim$  $\succ$  $\incomp$ & $\incomp$$\sim$$\incomp$$\incomp$ &   $\sim$  $\incomp$ & $\incomp$$\prec$$\sim$$\prec$ &   $\sim$  $\prec$  $\incomp$ & $\incomp$$\prec$$\sim$$\incomp$ &   $\sim$  $\prec$  $\incomp$ \\ 
$\incomp$$\prec$$\prec$$\prec$ &   $\sim$  $\prec$  $\incomp$ & $\incomp$$\prec$$\prec$$\incomp$ &   $\sim$  $\prec$  $\incomp$ & $\incomp$$\prec$$\succ$$\sim$ &   $\sim$  $\prec$  $\succ$  $\incomp$ & $\incomp$$\prec$$\succ$$\prec$ &   $\sim$  $\prec$  $\succ$  $\incomp$ \\ 
$\incomp$$\prec$$\succ$$\succ$ &   $\sim$  $\prec$  $\succ$  $\incomp$ & $\incomp$$\prec$$\succ$$\incomp$ &   $\sim$  $\prec$  $\succ$  $\incomp$ & $\incomp$$\prec$$\incomp$$\sim$ &   $\sim$  $\prec$  $\incomp$ & $\incomp$$\prec$$\incomp$$\prec$ &   $\sim$  $\prec$  $\incomp$ \\ 
$\incomp$$\prec$$\incomp$$\succ$ &   $\sim$  $\prec$  $\succ$  $\incomp$ & $\incomp$$\prec$$\incomp$$\incomp$ &   $\sim$  $\prec$  $\incomp$ & $\incomp$$\succ$$\sim$$\succ$ &   $\sim$  $\succ$  $\incomp$ & $\incomp$$\succ$$\sim$$\incomp$ &   $\sim$  $\succ$  $\incomp$ \\ 
$\incomp$$\succ$$\prec$$\sim$ &   $\sim$  $\prec$  $\succ$  $\incomp$ & $\incomp$$\succ$$\prec$$\prec$ &   $\sim$  $\prec$  $\succ$  $\incomp$ & $\incomp$$\succ$$\prec$$\succ$ &   $\sim$  $\prec$  $\succ$  $\incomp$ & $\incomp$$\succ$$\prec$$\incomp$ &   $\sim$  $\prec$  $\succ$  $\incomp$ \\ 
$\incomp$$\succ$$\succ$$\succ$ &   $\sim$  $\succ$  $\incomp$ & $\incomp$$\succ$$\succ$$\incomp$ &   $\sim$  $\succ$  $\incomp$ & $\incomp$$\succ$$\incomp$$\sim$ &   $\sim$  $\succ$  $\incomp$ & $\incomp$$\succ$$\incomp$$\prec$ &   $\sim$  $\prec$  $\succ$  $\incomp$ \\ 
$\incomp$$\succ$$\incomp$$\succ$ &   $\sim$  $\succ$  $\incomp$ & $\incomp$$\succ$$\incomp$$\incomp$ &   $\sim$  $\succ$  $\incomp$ & $\incomp$$\incomp$$\sim$$\sim$ &   $\sim$  $\incomp$ & $\incomp$$\incomp$$\sim$$\prec$ &   $\sim$  $\prec$  $\incomp$ \\ 
$\incomp$$\incomp$$\sim$$\succ$ &   $\sim$  $\succ$  $\incomp$ & $\incomp$$\incomp$$\sim$$\incomp$ &   $\sim$  $\incomp$ & $\incomp$$\incomp$$\prec$$\sim$ &   $\sim$  $\prec$  $\incomp$ & $\incomp$$\incomp$$\prec$$\prec$ &   $\sim$  $\prec$  $\incomp$ \\ 
$\incomp$$\incomp$$\prec$$\succ$ &   $\sim$  $\prec$  $\succ$  $\incomp$ & $\incomp$$\incomp$$\prec$$\incomp$ &   $\sim$  $\prec$  $\incomp$ & $\incomp$$\incomp$$\succ$$\sim$ &   $\sim$  $\succ$  $\incomp$ & $\incomp$$\incomp$$\succ$$\prec$ &   $\sim$  $\prec$  $\succ$  $\incomp$ \\ 
$\incomp$$\incomp$$\succ$$\succ$ &   $\sim$  $\succ$  $\incomp$ & $\incomp$$\incomp$$\succ$$\incomp$ &   $\sim$  $\succ$  $\incomp$ & $\incomp$$\incomp$$\incomp$$\sim$ &   $\sim$  $\incomp$ & $\incomp$$\incomp$$\incomp$$\prec$ &   $\sim$  $\prec$  $\incomp$ \\ 
$\incomp$$\incomp$$\incomp$$\succ$ &   $\sim$  $\succ$  $\incomp$ & $\incomp$$\incomp$$\incomp$$\incomp$ &   $\sim$  $\incomp$ & & \\
\end{tabular}
}
\end{center}
\caption{Extensive analysis of PDT\label{tab:table}}
\end{table}

%% file: section3.tex
\section{Properties of Partial Preferences} \label{sec:props}

In the following, given two distributions $f$ and $g$, 
we write $f\sim g$ for $f\preceq g$ and $g\preceq f$ (i.e., equivalence),
and we write $f \incomp g$ for $f \not\preceq g$ and $g \not\preceq f$ (i.e., incomparability).
First, we show that two relevant properties of classical Decision Theory 
continue to hold in PDT.  
Given two distributions $f,g\in\prob(\A)$, we write $f\rightarrow g$ in case there exist $\epsilon > 0$ 
and two alternatives $a_1$ and $a_2$ such that
\emph{(i)} $[a_1] \prec [a_2]$,
\emph{(ii)} $g(a_1)=f(a_1)-\epsilon$, $g(a_2)=f(a_2)+\epsilon$, and 
\emph{(iii)} for all $a\ne a_1,a_2$, $g(a)=f(a)$.
When $f \rightarrow g$, $g$ can be obtained from $f$ by shifting a positive amount of probability
from an alternative $a_1$ to a strictly preferred alternative $a_2$.
Then, denote by $f\Rightarrow g$ the transitive closure of $\rightarrow$.
\begin{theorem}\label{theo:shifting}
Let $f,g\in\prob(\A)$. If $f\Rightarrow g$ then $f\prec g$.  
\end{theorem}
\begin{proof}
It suffices to show that $f\rightarrow g$ implies $f\prec g$.
Assume that for some $a_1$ and $a_2$, where $[a_1] \prec [a_2]$, there exists $\epsilon>0$ such that $f(a_1)=g(a_1)+\epsilon$ and $f(a_2)= g(a_2)-\epsilon$.
Let $\gamma=g(a_1)+g(a_2)=f(a_1)+f(a_2)$, $\alpha=\frac{f(a_2)}{\gamma}$ and $\beta=\frac{g(a_2)}{\gamma}$. Then, $g$ can be written as $\conv{\gamma}{g'}{h}$ and $f$ as $\conv{\gamma}{f'}{h}$,
where $h$ is a probability distribution such that $h(a_1)=h(a_2)=0$, $g' = \conv{\beta}{[a_2]}{[a_1]}$ and 
$f' = \conv{\alpha}{[a_2]}{[a_1]}$. 
Since $[a_1] \prec [a_2]$ and $\alpha<\beta$, Axiom~\ref{ca:prec1} implies that $f'\prec g'$. 
Then, due to Axiom~\ref{ca:pref_dom}, $f\prec g$. 
\end{proof}
Another property that can be proved in PDT is that if each alternative coming out from a distribution $f$ is dominated by all the 
alternatives from $g$, then $f \preceq g$.
Preliminarily, given a distribution $f \in \prob(\A)$, the \emph{support} of $f$, $\supp(f)=\{a\in\A\mid f(a)>0 \}$, is the set of alternatives to which $f$ assigns positive probability.   
\begin{theorem}\label{theo:preceq}
Let $f,g\in\prob(\A)$ be such that, for all $a\in\supp(f)$ and $a' \in\supp(g)$, $[a] \preceq [a']$.
Then, $f \preceq g$.
\end{theorem}
\begin{proof}
We first show that for all $a\in \supp(f)$, $[a] \preceq g$. The proof is by induction
on the cardinality $n$ of $\supp(g)$ where the case $n=1$ is trivial.
Assume $n>1$, then $g$ can be written as $\conv{\alpha}{[a']}{g'}$, where $a'$ is a generic alternative from $\supp(g)$ and $g'(a')=0$.
Clearly, the distribution $[a]$ can also be written as $\conv{\alpha}{[a]}{[a]}$.
By assumption $a\preceq a'$ and, since the cardinality of $g'$ is $n-1$, by induction $[a]\preceq g$. Then, by applying Axiom~\ref{ca:preceq1} we have $[a]\preceq g$. 

Now, the proof is by induction on the cardinality $m$ of $\supp(f)$ where the base case $m=1$ has been proved above. Assume $m>1$ and let
$a\in \supp(f)$,
then $f$ can be written as $\conv{\alpha}{[a]}{f_{-a}}$, where  $f_{-a}(a)=0$ and $f_{-a}(a')= \frac{f(a')}{1-f(a)}$ for all $a' \in \supp(f)\setminus\{a\}$.  
We already proved that $[a] \preceq g$ and by induction hypothesis it holds also that $f_{-a}\preceq g$. Then, by writing $g$ as
$\conv{\alpha}{g}{g}$ and applying Axiom~\ref{ca:preceq1}  we have that $f\preceq g$.
\end{proof}
From Theorem~\ref{theo:preceq} it is immediate to show that distributions over equivalent alternatives are equivalent themselves.
What if some of the alternatives are incomparable? 
We show that their distributions are either equivalent or incomparable, 
as due to Axiom~\ref{ax:upper1} no strict preference can be derived.

\begin{lemma}\label{lem:incomparable_not_prec}
Let $f,g\in\prob(\A)$ be such that, for all $a\in\supp(f)$ and $a' \in\supp(g)$, $a \incomp a'$.
Then, either $f \incomp g$ or $f \sim g$.
\end{lemma}
\begin{proof}
 Let $n_f = |\supp(f)|$ and $n_g = |\supp(g)|$,
 we proceed by induction on $n_f+n_g$.
 If $n_f+n_g = 2$, the thesis is obviously true.
 Otherwise, assume w.l.o.g.\ that $n_f>1$ and let $a \in \supp(f)$.
 We can write $f$ as $\conv{f(a)}{[a]}{f_{-a}}$,
 where $f_{-a}(a) = 0$ and $f_{-a}(a') = \frac{f(a')}{1-f(a)}$ for all $a' \in \supp(f)\setminus\{a\}$.
 Since the support of $f_{-a}$ is smaller than the one of $f$, we can apply the inductive
 hypothesis to the pair $f_{-a},g$, obtaining that either $f_{-a} \incomp g$ or $f_{-a} \sim g$.
 We can also apply the inductive hypothesis to the pair $[a],g$, obtaining that 
 $[a] \incomp g$ or $[a] \sim g$.
 Assume by contradiction that  $f \bowtie g$ for some $\bowtie \in \{\prec, \succ\}$.
 By Axiom~\ref{ax:upper1}, it holds $[a] \bowtie g$ or $f_{-a} \bowtie g$, 
 which is a contradiction.
\end{proof}
Finally, the following generalization of Lemma~\ref{lem:incomparable_not_prec} shows that
alternatives that are either incomparable or equivalent lead to distributions
that are themselves either incomparable or equivalent.
\begin{theorem}\label{theo:equivncomp_prec}
Let $f,g\in\prob(\A)$ be such that, for all $a\in\supp(f)$ and $a' \in\supp(g)$, either $a \sim a'$ or $a \incomp a'$.
Then, $f \incomp g$ or $f \sim g$.
\end{theorem}
\begin{proof}
The proof is very similar to the one of Lemma~\ref{lem:incomparable_not_prec},
where only the base case is affected by the weakened assumption.
\end{proof}

%% file: conclusions.tex
\section{Conclusions}

In this work we challenged the customary decision-theoretic assumption of totality of the
preference relations, on the basis of two real-world scenarios.
We proposed a weakening of the classical theory and proved that it retains several
desirable properties, while allowing for incomparable alternatives.

Partial preferences have been already employed in procurement auctions.
In \cite{DBLP:conf/esorics/BonattiFGS11}  second-price auctions have been generalized by considering a bid domain representing information disclosures and two ad hoc partial preference relations have been defined for modeling the sensitivity of the disclosed data.  Then, in \cite{mfcs13-auctions} the previous framework has been extended for modeling also service cost, QoS and  functional differences, etc.
Somewhat surprisingly, it has been shown that extending second price auctions to partial preferences
does not yield truthful mechanisms, since overbidding may be profitable in some contexts. 
This means that, in general, partial preferences may significantly change the theoretical properties of a mechanism and
the axiomatization presented in this work enables to uniformly employ them in the field of Mechanism Design and estimate their impact. 